\documentclass{amsart}
\usepackage{amsmath,amsfonts,amsthm,amssymb,setspace}
\newtheorem{thm}{Theorem}[section]
\newtheorem*{thm*}{Theorem}

\newtheorem{lemma}[thm]{Lemma}

\newtheorem{coro}[thm]{Corollary}
\newtheorem{prop}[thm]{Proposition}
\newtheorem{defn}[thm]{Definition}
\numberwithin{equation}{section}

\begin{document}

\title{2D Locus Configurations and the Charged Trigonometric Calogero-Moser System}

\author{Greg Muller}
\address{Department of Mathematics,
Louisiana State University, Baton Rouge, LA 70808, USA}
\email{gmuller@lsu.edu}

\begin{abstract}
A central hyperplane arrangement in $\mathbb{C}^2$ with multiplicity is called a `locus configuration' if it satisfies a series of `locus equations' on each hyperplane.  Following \cite{CFV99}, we demonstrate that the first locus equation for each hyperplane corresponds to a force-balancing equation on a related interacting particle system on $\mathbb{C}^*$: the charged trigonometric Calogero-Moser system.  When the particles lie on $S^1\subset \mathbb{C}^*$, there is a unique equilibrium for this system.  For certain classes of particle weight, this is enough to show that all the locus equations are satisfied, producing explicit examples of real locus configurations.  This in turn produces new examples of Schr\"odinger operators with Baker-Akhiezer functions.
\end{abstract}

\maketitle

\section{Introduction.}

\subsection{Locus Equations and Configurations.}
Let $\mathcal{A}$ be a central hyperplane arrangement in $\mathbb{C}^r$ with a positive integer multiplicity $m_i$ assigned to every hyperplane $H_i$.  Define the \emph{potential} $u(x)$ on $\mathbb{C}^r$
\[ u(x):=\sum_{i} \frac{m_i(m_i+1)\langle \alpha_i,\alpha_i\rangle}{\langle \alpha_i,x\rangle^2}\]
where for all $i$, $\alpha_i$ is a normal vector to $H_i$.\footnote{Here, brackets denote the natural $\mathbb{C}$-linear inner product on $\mathbb{C}^r$.}  In \cite{CFV99}, Chalykh, Feigin and Veselov investigate the associated Schr\"odinger operator $L:=\Delta+u(x)$ (where $\Delta$ is the Laplacian on $\mathbb{C}^r$); specifically, the question of when $L$ has a \emph{Baker-Akhiezer function} (BA function).

Chalykh, Feigin and Veselov show that such a $u(x)$ has a Baker-Akhiezer function if and only if $\mathcal{A}$ satisfies a series of equations called the \emph{locus equations}; in such a case $\mathcal{A}$ is called a \emph{locus configuration}.  In this paper, we focus on the geometry of $\mathcal{A}$ and locus configurations, rather than on the BA function $\psi$ and the operator $L$.

We consider a subset of the locus equations, called the \emph{first-locus equations} (so-called because it contains only the first equation in each series); an arrangement satisfying them will be called a \emph{first-locus configuration}.  The significance of this concept is that it can be shown to be equivalent to an interacting particle system being in equilibrium.

\subsection{The Charged Trigonometric Calogero-Moser System.}

The charged trig. Calogero-Moser system describes a collection of $n$ particles with position $(\theta_1,\theta_2,...\theta_n)$ and `charges' $(q_1,q_2,...q_n)$ with interaction potential given by:
\[ \sum_{\substack{1\leq i,j \leq n\\i\neq j}}\frac{q_iq_j}{\sin^2\left(\frac{\theta_j-\theta_i}{2}\right)}\]
This describes $n$ charged particles on the circle repelling each other proportional to the inverse cube of their separation.

Given a central arrangement $\mathcal{A}$ in $\mathbb{C}^2$ with multiplicity, a corresponding ensemble of charged particles $E(\mathcal{A})$ can be constructed (see Section \ref{section: CTCM}).
\begin{thm*}[\textbf{\ref{thm:main1}}] The arrangement $\mathcal{A}$ is a first-locus configuration if and only if the ensemble $E(\mathcal{A})$ is in equilibrium for the charged trig. CM system.
\end{thm*} 
This perspective of the locus equations was known to Chalykh, Feigin and Veselov in \cite{CFV99}; it is used in the proof of Theorem 4.5, and in the choice of the name `locus configuration'.  However, this fact deserves be to said explicitly.

As a consequence of this physical perspective, an existence and uniqueness theorem can be proven for real arrangements $\mathcal{A}$ which satisfy the first-locus equations.  From this, it follows that such configurations with symmetry in the set of multiplicities must possess corresponding reflection symmetries.  These symmetries, together with the first-locus equations, are enough to produce arrangements $\mathcal{A}$ which satisfy all of the locus equations.
\begin{thm*}[\textbf{\ref{thm:main2}}]
Let $\mathbf{m}=(m_1,m_2,...m_n)$ be a list of positive integers, such that for every $i$ such that $m_i>1$, and for all $j$, $m_{i+j}=m_{i-j}$ (indices are mod $n$).  Then there exists a real 2D locus configuration $\mathcal{A}_{\mathbf{m}}$ with cyclically ordered multiplicities $\mathbf{m}$, which is unique up to rotation of $\mathcal{A}_{\mathbf{m}}$.
\end{thm*}

\subsection{Acknowledgements.}
This paper most immediately owes its existence to \cite{CFV99} and its authors.  The author is personally grateful to Yuri Berest and Oleg Chalykh for many conversations about the subject matter, as well as for support and advice.  

%
%
%

\section{Hyperplane Arrangements and Locus Configurations.}

Let $\mathcal{A}$ be a central hyperplane arrangement in $\mathbb{C}^r$ consisting of hyperplanes $H_i$, together with multiplicities $m_i\in \mathbb{N}$ for every $H_i$.\footnote{Throughout, we assume that the the arrangement $\mathcal{A}$ contains no degenerate hyperplanes; that is, that $\forall i$, $\langle\alpha_i,\alpha_i\rangle\neq 0$.}  For each hyperplane $H_i$, choose a normal vector $\alpha_i$.  Then, define the function
\[ u(x):=\sum_{i} \frac{m_i(m_i+1)\langle \alpha_i,\alpha_i\rangle}{\langle \alpha_i,x\rangle^2}\]
Note that rescaling any of the $\alpha_i$ by a non-zero constant fixes $u$, and so $u$ is independent of the choice of $\alpha_i$.  

The arrangement $\mathcal{A}$ is called a \textbf{locus configuration} if, for each hyperplane $H_i$ and each integer $k\in 1,2,... m_i$, the function
\[ \sum_{ j| j\neq i}\frac{m_j(m_j+1)\langle \alpha_j,\alpha_j\rangle \langle \alpha_j,\alpha_i\rangle^{2k-1}}{\langle \alpha_j,x\rangle^{2k+1}}=0\]
for all $x$ in the hyperplane $H_i$; the above equation is called \textbf{$k$th locus equation at $H_i$}.  Again, this property is intrinsic to the arrangement and not the scaling of $\alpha$.

The significance of locus configurations comes from the following theorem.
\begin{thm}[\cite{CFV99} Theorem 3.1.]
Let $\mathcal{A}$ and $u(x)$ be as above.  Then the operator $L=\Delta+u(x)$ has a Baker-Akhiezer function if and only if $\mathcal{A}$ is a locus configuration.
\end{thm}
One-dimensional Baker-Akhiezer functions were introduced by Krichever in \cite{Kri77}, and generalized to multiple dimensions in \cite{CFV99}.  BA functions have applications to the bispectral problem, Darboux factorization and the study of Huygens' principle (see \cite{DG86},\cite{VSC93},\cite{BV94},\cite{BV98},\cite{Ber98}).  However, this paper will neither define nor use the BA function produced by the theorem; we will instead focus on constructing locus configurations.

Locus configurations in $\mathbb{C}^2$ are of elementary interest, because all locus configurations are assembled out of two-dimensional ones.
\begin{thm}[\cite{CFV99} Theorem 4.1]
An arrangement $\mathcal{A}$ is a locus configuration if and only if the restriction of $\mathcal{A}$ to any two-dimensional subsystem is a locus configuration.  That is, for any two plane $\pi\subset \mathbb{C}^n$, the vectors $\alpha_i\in \mathcal{A}\cap\pi$ with their multiplicities $m_i$ must satisfy the locus equations.
\end{thm}

\section{Symmetry and First-Locus Configurations.}

The first examples of locus configurations are \textbf{Coxeter arrangements}, which are arrangements $\mathcal{A}$ such that for any hyperplane $H\in \mathcal{A}$, the reflection across $H$ leaves $\mathcal{A}$ invariant\footnote{That is, reflection across $H$ sends hyperplanes in $\mathcal{A}$ to hyperplanes in $\mathcal{A}$ with the same multiplicity}.  That these are locus configurations is a consequence of the following lemma.
\begin{lemma}\label{lemma:ref}
Let $\mathcal{A}$ be an arrangement, and $H_i$ a hyperplane in $\mathcal{A}$.  If reflection across $H_i$ leaves $\mathcal{A}$ invariant, then all the locus equations at $H_i$ are satisfied.
\end{lemma}
\begin{proof}
Let $R_i$ denote the reflection across $H_i$.  Then for any other $H_j$, $j\neq i$, consider the terms in the $k$th locus equation at $H_i$ corresponding to $H_j$ and $R_iH_j$.
\[ \frac{m_j(m_j+1)\langle \alpha_j,\alpha_j\rangle \langle \alpha_j,\alpha_i\rangle^{2k-1}}{\langle \alpha_j,x\rangle^{2k+1}}+\frac{m_j(m_j+1)\langle R_i\alpha_j,R_i\alpha_j\rangle \langle R_i\alpha_j,\alpha_i\rangle^{2k-1}}{\langle R_i\alpha_j,x\rangle^{2k+1}}\]
Since reflection is an isometry, $\langle R_i\alpha_j,R_i\alpha_j\rangle=\langle\alpha_j,\alpha_j\rangle$.  Furthermore, since $x\in H_i$, $\langle R_i\alpha_j,x\rangle=\langle \alpha_j,x\rangle$.  Finally, $\langle R_i\alpha_j,\alpha_i\rangle=-\langle \alpha_j,\alpha_i\rangle$, and so the two terms in the $k$th locus equation corresponding to $H_j$ and $R_iH_j$ cancel out.  Since all the terms can be paired up in this manner, the total sum is zero.
\end{proof}

Chalykh, Feigin and Veselov show that for any locus configuration $\mathcal{A}$, any hyperplanes of so-called `large multiplicity' must have the property that reflection across them leaves $\mathcal{A}$ invariant.  These are the hyperplanes $H_i$ whose multiplicity $m_i$ is greater than the largest number of hyperplanes simultaneously intersecting $H_i$ in a codimension 1 subspace.  If the arrangement is 2-dimensional, then this is one less than the number of hyperplanes in $\mathcal{A}$.

This paper will investigate the extreme of this condition, when there is a reflection symmetry at any hyperplane of multiplicity $>1$.
\begin{defn}
An arrangement $\mathcal{A}$ is called \textbf{coarsely Coxeter} if, for any $H_i\in \mathcal{A}$ such that $m_i>1$, reflection across $H_i$ leaves $\mathcal{A}$ invariant.
\end{defn}
By the lemma, such an arrangement will satisfy the locus equations at any hyperplane of multiplicity $>1$.  Therefore, to be a locus configuration, the locus equations need only be checked at hyperplanes of multiplicity $1$, where there is only a single locus equation.  The remainder of the note is concerned with producing examples of coarsely Coxeter arrangements which are locus configurations.

The only meaningful equations to check for a coarsely Coxeter arrangement are the first locus equations.  Therefore, we introduce the following weaker condition on an arrangement.
\begin{defn}
An arrangement $\mathcal{A}$ is called a \textbf{first-locus configuration} if for every hyperplane $H_i\in \mathcal{A}$, the first locus equation is satisfied at $H_i$; that is, $\forall x\in H_i$,
\[ \sum_{ j| j\neq i}\frac{m_j(m_j+1)\langle \alpha_j,\alpha_j\rangle \langle \alpha_j,\alpha_i\rangle}{\langle \alpha_j,x\rangle^{3}}=0\]
\end{defn}
The usefulness of this concept is the following.
\begin{prop}\label{prop:locus}
Any coarsely Coxeter arrangement which is a first-locus configuration is necessarily a locus configuration.
\end{prop}
\begin{proof}
At every hyperplane of multiplicity 1, the locus equations are satisfied by the first-locus condition.  At any hyperplane of multiplicity $>1$, the locus equations are satisfied by the coarsely Coxeter condition and Lemma \ref{lemma:ref}.
\end{proof}


\section{The Charged Trigonometric Calogero-Moser System.}\label{section: CTCM}

From now on, we reduce the scope of our investigation to arrangements in $\mathbb{C}^2$.  This is a significant simplification, because each hyperplane $H_i$ is a complex line, and can now be spanned by a single vector $\alpha_i^\perp$.  The locus equations then need only be checked to vanish at $\alpha_i^\perp$.

Instead of thinking of a collection of lines $\mathcal{A}$ in $\mathbb{C}^2$, we can think of a collection of points $E(\mathcal{A})$ in $\mathbb{C}/2\pi\mathbb{Z}$, the space of non-isotropic lines in $\mathbb{C}^2$.  Explicitly, for every $i$, there is a unique $\theta_i\in \mathbb{C}/2\pi\mathbb{Z}$ such that $\alpha_i$ is a multiple of $(\cos(\theta_i/2),\sin(\theta_i/2))$.  For simplicity, we assume that $\alpha_i=(\cos(\theta_i/2),\sin(\theta_i/2))$, and we define $\alpha_i^\perp=(-\sin(\theta_i/2),\cos(\theta_i/2))$, so that $\alpha_i^\perp$ spans $H_i$.  Furthermore, let $q_i=m_i(m_i+1)$.

%
%

In this notation, the first locus equation at $H_i$ becomes
\[0=\sum_{j|j\neq i}\frac{m_j(m_j+1)\langle \alpha_i,\alpha_i\rangle \langle \alpha_i,\alpha_j\rangle }{\langle\alpha_i^\perp,\alpha_j\rangle^3}=\sum_{j|j\neq i} \frac{q_j\cos\left(\frac{\theta_j-\theta_i}{2}\right)}{\sin^3\left(\frac{\theta_j-\theta_i}{2}\right)}\]
The idea now is to interpret the $j$th term of this sum as \emph{`the force a particle at $\theta_j$ with charge $q_i$ exerts on a particle at $\theta_i$'} for some particle interaction.  This will make the above sum a force balancing equation for a particle at $\theta_i$, and the system of first-locus equations will become the requirement that the particle ensemble $(\theta_1,\theta_2,...\theta_n)$ is in equilibrium.

Let $E$ be a collection of distinct points $(\theta_1,\theta_2,...\theta_n)\in \mathbb{C}/2\pi\mathbb{Z}$, together with positive integers $(q_1,q_2,...q_n)$ called the \textbf{charges}.  Define the \textbf{charged trigonometric Calogero-Moser potential} of $E$ to be
\[ \mu(E)=\sum_{\substack{i,j\\i\neq j}}\frac{q_iq_j}{\sin^2\left(\frac{\theta_j-\theta_i}{2}\right)}\]
In the case that $q_1=q_2=...=q_n=1$, this is the usual trigonometric Calogero-Moser potential.  The force acting on the particle at $\theta_i$ is then given by
\[ \frac{\partial\mu}{\partial \theta_i}(E)=q_i\sum_{j|j\neq i}\frac{q_j\cos\left(\frac{\theta_j-\theta_i}{2}\right)}{\sin^3\left(\frac{\theta_j-\theta_i}{2}\right)}\]
Since $q_i$ is a non-zero constant, this gives the desired force balancing condition.  The system $E$ is in equilibrium if all the partial derivatives $\frac{\partial\mu}{\partial\theta_i}(E)$ vanish.  Therefore,
\begin{thm}\label{thm:main1}
Let $\mathcal{A}$ be an arrangement in $\mathbb{C}^2$.  Then TFAE.
\begin{itemize}
\item $\mathcal{A}$ is a first-locus configuration.
\item The collection of particles $E(\mathcal{A})$ with particles at $(\theta_1,\theta_2,...\theta_n)$ and with charges $(m_1(m_1+1),m_2(m_2+1),...m_n(m_n+1))$ is in equilibrium for the charged trig. CM potential $\mu(E(\mathcal{A}))$.
\item $E(\mathcal{A})$ is a critical point for the potential $\mu$.
\end{itemize}
\end{thm}

\section{Existence and Uniqueness of Real Equilibria.}

Let $\mathcal{A}$ be a real arrangement in $\mathbb{C}^2$; that is, the $\alpha_i$ may be chosen to have real coordinates.  This implies that the $\theta_i$ are also real, and so $\theta_i\in \mathbb{R}/2\pi\mathbb{Z}\simeq S^1$.  The charged trigonometric CM potential then describes a repelling force between pairs of charged particles on the circle. The dynamical perspective gives helpful intuition to the problem of finding equilibria, because a collection of repelling particles on a compact space should trend toward some equilibrium (provided there is some dampening effect).  We now make this argument precise.

We say that a collection of particles $E$ on $S^1$ is \textbf{cyclically ordered} if the $\theta_i$ occur in the correct cyclic order; that is, if there is some $i$ such that
\[\theta_i<\theta_{i+1}<...<\theta_n<\theta_1<...\theta_{i-1}\]
for representatives of the $\theta_i$ in $[0,2\pi)$.

\begin{thm}
Let $\mathbf{q}=(q_1,q_2,...q_n)$ be a list of positive integers.  Then there exists a cyclically ordered $E=(\theta_1,\theta_2,...\theta_n), \theta_i\in \mathbb{R}/2\pi\mathbb{Z}$ with charge $\mathbf{q}$ such that $E$ is a critical point of $\mu$; that is, that $E$ is in equilibrium for the charged trig. CM potential.  Furthermore, this $E$ is unique up to simultaneous rotation of the system.
\end{thm}
\begin{proof}
Let $X$ be the space of all cyclically ordered $(\theta_1,\theta_2,...\theta_n)$ in $\mathbb{R}/2\pi\mathbb{Z}$.  It is a connected component of $(\mathbb{R}/2\pi\mathbb{Z})^n$ minus the `fat diagonal', those points where any two of the coordinates coincide.  $X$ is convex, in the sense that for any two points $E, E'\in X$, it contains the straight line $tE+(1-t)E', 0\leq t\leq 1$ connecting them.

Let $U(\theta)=\sin^{-2}(\theta/2)$ be the pairwise trig. CM potential.  The important facts about $U$ are: it is strictly convex on $(0,2\pi)$, it is bounded below, and it approaches $+\infty$ at either boundary.  The charged trig. CM potential $\mu$ can be written in terms of $U$.
\[\mu(E)=\sum_{\substack{i,j\\i\neq j}}q_iq_jU(\theta_j-\theta_i)\]
From the form of this equation, it is clear that $\mu$ is convex and bounded-below on $X$ (though no longer strictly convex), and that $\mu(E)$ approaches $+\infty$ as $E$ approaches the boundary of $X$.  Because $\mu$ is convex on a convex domain, its set of critical points is a convex subset of $X$.  Because $\mu$ is bounded below and it approaches $+\infty$ on the boundary, it has an absolute minimum; by the previous remark, then every critical point is an absolute minimum.

Let $E$ and $E'$ both be absolute minima of $\mu$ on $X$, and assume that there are $i$ and $j$ such that $\theta_j-\theta_i\neq \theta'_j-\theta'_i$.  On the straight line connecting $E$ and $E'$, $U(\theta_j-\theta_i)$ is strictly convex because the argument is non-constant, and so its value on the interior is strictly less than the value at the endpoints.  However, because the other terms in $\mu$ are also convex, the value of $\mu$ at any interior point on the line is strictly less than $\mu(E)$ and $\mu(E')$, which contradicts the $E$ and $E'$ being absolute minima.  Therefore, $\theta_j-\theta_i=\theta'_j-\theta'_i$ for all $i,j$, and so $E$ and $E'$ differ by a rotation.
\end{proof}

\begin{coro}
Let $\mathbf{m}=(m_1,m_2,...m_n)$ be a list of positive integers.  Then there is a real first-locus configuration $\mathcal{A}_{\mathbf{m}}$ in $\mathbb{C}^2$ whose cyclic list of multiplicities is $\mathbf{m}$, and $\mathcal{A}_{\mathbf{m}}$ is unique up to rotation of the system.
\end{coro}

From now on, $\mathcal{A}_{\mathbf{m}}$ will denote this arrangement.


\section{Real Locus Configurations.}

By the previous section, the existence and uniqueness of real 2D first-locus configurations is clear.  Therefore, to show $\mathcal{A}_{\mathbf{m}}$ is a locus configuration, it is sufficient to show that $\mathcal{A}_{\mathbf{m}}$ is coarsely Coxeter.  Fortunately, the uniqueness of $\mathcal{A}_{\mathbf{m}}$ means that symmetries in the multiplicities $\mathbf{m}$ will imply symmetries in $\mathcal{A}_{\mathbf{m}}$.

\begin{lemma}
Let $\mathbf{m}=(m_1,m_2,...m_n)$ be a list of positive integers.  Let $i$ be such that, for all $j$, $m_{i+j}=m_{i-j}$.\footnote{The indices on $m$ on are taken mod $n$ here and afterward.}  Then reflection across the hyperplane $H_i$ leaves the arrangement $A_{\mathbf{m}}$ invariant.
\end{lemma}
\begin{proof}
Let $\mathbf{m}'$ denote the list of integers such that $m'_{i+j}=m_{i-j}$.  Let $R_i\mathcal{A}_{\mathbf{m}}$ denote the reflection of the arrangement $\mathcal{A}_{\mathbf{m}}$ across the hyperplane $H_i$.  Being a first-configuration is preserved by reflection, and so $R_i\mathcal{A}_{\mathbf{m}}$ is a 2D real first-locus configuration, with cyclically ordered multiplicities $\mathbf{m}'$.  By the uniqueness theorem, $R_i\mathcal{A}_{\mathbf{m}}$ is a rotation of $\mathcal{A}_{\mathbf{m}'}$.  However, because $H_i$ is fixed by $R_i$, this rotation must be the identity, and so $R_i\mathcal{A}_{\mathbf{m}}=\mathcal{A}_{\mathbf{m}'}$.
\end{proof}

This means whether or not $\mathcal{A}_{\mathbf{m}}$ is coarsely Coxeter can be read off from $\mathbf{m}$.

\begin{defn}
Let $\mathbf{m}=(m_1,m_2,...m_n)$ be a list of positive integers.  $\mathbf{m}$ is \textbf{coarsely symmetric} if for every $i$ such that $m_i>1$, and for all $j$, $m_{i+j}=m_{i-j}$.
\end{defn}
The following lemma is then immediate from the definition and the prior lemma.
\begin{lemma}\label{lemma:coarse}
Let $\mathbf{m}$ be a list of positive integers.  Then $\mathcal{A}_{\mathbf{m}}$ is coarsely Coxeter if and only if $\mathbf{m}$ is coarsely symmetric.
\end{lemma}
Finally, we arrive at the main result of this section.

\begin{thm}\label{thm:main2}
Let $\mathbf{m}$ be a list of positive integers, which is coarsely symmetric.  Then there exists a real 2D locus configuration $\mathcal{A}_{\mathbf{m}}$ with cyclically ordered multiplicities $\mathbf{m}$, which is unique up to rotation of $\mathcal{A}_{\mathbf{m}}$.
\end{thm}
\begin{proof}
By the construction of $\mathcal{A}_{\mathbf{m}}$, it is a first-locus configuration.  By Lemma \ref{lemma:coarse}, $\mathcal{A}_{\mathbf{m}}$ is coarsely Coxeter.  Thus, by Proposition \ref{prop:locus}, $\mathcal{A}_{\mathbf{m}}$ is a locus configuration.
\end{proof}

\section{Examples and Questions.}

All real 2D locus configurations in \cite{CFV99} or known to the author are of the form $\mathcal{A}_{\mathbf{m}}$, for $\mathbf{m}$ coarsely symmetric.  These examples are:
\begin{itemize}
\item All 2D Coxeter arrangements arise for an $\mathbf{m}$ which is completely symmetric (that is, $m_{i+j}=m_{i-j}$ for all $i,j$).
\item The arrangements $A_2(m)$ (see \cite{CFV99}, pg. 13) correspond to $\mathbf{m}=(m,1,1)$.
\item The arrangements $C_2(m,l)$ (see \cite{CFV99},pg. 14) correspond to $\mathbf{m}=(m,1,l,1)$.
\end{itemize}
However, Theorem \ref{thm:main2} guarantees the existence of other 2D real locus configurations that have not appeared in the literature.  For example, $\mathbf{m}=(m,1,1,...1)$ for any number of $1$s will correspond to a locus configuration.

We close with a couple of natural questions.
\begin{itemize}
\item \textbf{Is every 2D real locus configuration coarsely Coxeter?}  
    This would mean that every real locus configuration could be obtained by starting with a Coxeter configuration, and adding multiplicity 1 hyperplanes in a symmetric way.  
\item \textbf{Is there a formula for producing $\mathcal{A}_{\mathbf{m}}$ given $\mathbf{m}$?} The dynamic perspective means that such arrangements can be approximated by techniques like Newton's method.  However, because many of the properties of locus configurations and BA functions require exact formula.
\end{itemize}
%
%

\bibliography{MyNewBib}{}

\def\cprime{$'$} \def\cprime{$'$} \def\cprime{$'$} \def\cprime{$'$}
  \def\cprime{$'$}
\providecommand{\bysame}{\leavevmode\hbox to3em{\hrulefill}\thinspace}
\providecommand{\MR}{\relax\ifhmode\unskip\space\fi MR }
\providecommand{\MRhref}[2]{%
  \href{http://www.ams.org/mathscinet-getitem?mr=#1}{#2}
}
\providecommand{\href}[2]{#2}
\begin{thebibliography}{CFV99}

\bibitem[Ber98]{Ber98}
Yuri Berest, \emph{Hierarchies of {H}uygens' operators and {H}adamard's
  conjecture}, Acta Appl. Math. \textbf{53} (1998), no.~2, 125--185.
  \MR{MR1646583 (99j:58204)}

\bibitem[BV94]{BV94}
Yuri Berest and A.~P. Veselov, \emph{The {H}adamard problem and {C}oxeter
  groups: new examples of the {H}uygens equations}, Funktsional. Anal. i
  Prilozhen. \textbf{28} (1994), no.~1, 3--15, 95. \MR{MR1275723 (95h:58131)}

\bibitem[BV98]{BV98}
\bysame, \emph{Singularities of the potentials of exactly solvable
  {S}chr\"odinger equations, and the {H}adamard problem}, Uspekhi Mat. Nauk
  \textbf{53} (1998), no.~1(319), 211--212. \MR{MR1618173}

\bibitem[CFV99]{CFV99}
O.~A. Chalykh, M.~V. Feigin, and A.~P. Veselov, \emph{Multidimensional
  {B}aker-{A}khiezer functions and {H}uygens' principle}, Comm. Math. Phys.
  \textbf{206} (1999), no.~3, 533--566. \MR{MR1721907 (2001g:37131)}

\bibitem[DG86]{DG86}
J.~J. Duistermaat and F.~A. Gr{\"u}nbaum, \emph{Differential equations in the
  spectral parameter}, Comm. Math. Phys. \textbf{103} (1986), no.~2, 177--240.
  \MR{MR826863 (88j:58106)}

\bibitem[Kri77]{Kri77}
I.M. Krichever, \emph{Methods of algebraic geometry in the theory of nonlinear
  equations}, Uspekhi Mat. Nauk \textbf{v.32(6)} (1977), p.198–245.

\bibitem[VSC93]{VSC93}
A.~P. Veselov, K.~L. Styrkas, and O.~A. Chalykh, \emph{Algebraic integrability
  for the {S}chr\"odinger equation, and groups generated by reflections},
  Teoret. Mat. Fiz. \textbf{94} (1993), no.~2, 253--275. \MR{MR1221735
  (94j:35151)}

\end{thebibliography}
\bibliographystyle{amsalpha}

\end{document}